\documentclass{amsart}

\usepackage{amsmath,amsthm,amsfonts,amssymb,color}
\usepackage{epsfig}
\usepackage[ruled,vlined]{algorithm2e}
\usepackage{multirow}

\newcommand{\tabincell}[2]{\begin{tabular}{@{}#1@{}}#2\end{tabular}}
\usepackage{array}
\usepackage{multirow}
\usepackage{epstopdf}
\usepackage{relsize}

\usepackage{cite}
\usepackage{pmat}
\usepackage{epstopdf}
\usepackage{flushend}
\usepackage{diagbox}
\usepackage{scalerel}
\usepackage{verbatim}

\usepackage[utf8]{inputenc}

\newtheorem{theorem}{Theorem}

\newtheorem{lemma}{Lemma}
\newtheorem{remark}{Remark}
\newtheorem{definition}{Definition}
\newtheorem{proposition}{Proposition}

\newtheorem{assumption}{Assumption}

\newcolumntype{y}[1]{>{\centering\let\newline\\\arraybackslash\hspace{0pt}}p{#1}}

\makeatletter
\g@addto@macro{\endabstract}{\@setabstract}
\newcommand{\authorfootnotes}{\renewcommand\thefootnote{\@fnsymbol\c@footnote}}%
\makeatother
\begin{document}
\begin{center}
  \LARGE 
  Decentralized Secure State-Tracking \\in Multi-Agent Systems\footnote{This work was funded in part by the Army Research Laboratory under Cooperative Agreement W911NF-17-2-0196 and in part by NSF grant 1705135.} \par \bigskip

  \normalsize
  \authorfootnotes
  Yanwen Mao and Paulo Tabuada \par \bigskip

  Department of Electrical and Computer Engineering, \\ University of California, Los Angeles \\ Email: yanwen.mao@g.ucla.edu, tabuada@ucla.edu \par
 \bigskip

  \today
\end{center}

\begin{abstract}
This paper addresses the problem of decentralized state-tracking in the presence of sensor attacks. We consider a network of nodes where each node has the objective of tracking the state of a linear dynamical system based on its measurements and messages exchanged with neighboring nodes notwithstanding some measurements being spoofed by an adversary. We propose a novel decentralized attack-resilient state-tracking algorithm based on the simple observation that a compressed version of all the network measurements suffices to reconstruct the state. This motivates a 2-step solution to the decentralized secure state-tracking problem: (1) each node tracks the compressed version of all the network measurements, and (2) each node asymptotically reconstructs the state from the output of step (1). We prove that, under mild technical assumptions, our algorithm enables each node to track the state of the linear system and thus solves the decentralized secure state-tracking problem. 
\end{abstract}


\section{Introduction}
\label{sec:intro}

Over the past decade the topic of decentralized state-tracking has received considerable attention, due to the increasingly decentralized nature of complex systems such as traffic networks and power grids. In this problem, a group of nodes is required to collectively track the state of a linear dynamical system using measurements from their own sensors and messages exchanged with neighboring nodes via a communication network. 

In the absence of attacks, the decentralized state-tracking problem has been well studied~\cite{han2018simple,kim2016distributed,wang2017distributed,mitra2018distributed}. However, reports on cyber-physical systems attacks~\cite{farwell2011stuxnet,WinNT} remind us of how vulnerable these systems can be. Motivated by this consideration, in this paper we tackle a more challenging version of the decentralized state-tracking problem where some nodes are subject to sensor attacks spoofing its measurements. We refer to this problem as the Decentralized Secure State-Tracking (DSST) problem.

\subsection{Related Work}
Closely related to the DSST problem is the Secure State-Reconstruction (SSR) problem whose formal definition can be found in~\cite{mao19}. Roughly speaking, in the SSR problem a central server, which has access to all measurements, is asked to reconstruct the state of a linear system despite an attack on some sensors. A preliminary version of the SSR problem was first considered in~\cite{sandberg2010security}, and rigorously defined in~\cite{fawzi2011secure}. Moreover, the solvability of the SSR problem was settled in~\cite{fawzi2011secure,shoukry2015event,chong2015observability}. In particular, in~\cite{fawzi2011secure} and~\cite{chong2015observability} it is pointed out that in order to reconstruct the state in the presence of $s$ attacked sensors, the linear system must remain observable after a removal of any subset of $2s$ sensors. This property of a linear system is referred to as $2s-$sparse observability in~\cite{shoukry2015event}. 

Although it has been known for long that the SSR problem is NP-hard~\cite{Hamza14}, much progress 
was reported on reducing the computational complexity of solving the SSR problem. Many works, such as \cite{Hamza14,shoukry2015event,mitra2019byzantine,mao2021computational}, carve out subsets of the SSR problem instances which allow for a polynomial-time solution. More results on reducing the computational complexity of the SSR problem can be found in \cite{yong2016robust} and \cite{han2019convex}. In particular, the solution of the SSR problem proposed in~\cite{shoukry2018smt}, which is built upon satisfiability modulo theory, is of the utmost practical interest due to its good performance in simulations. We show in this paper that, with our algorithm, any DSST problem can be reduced to an SSR problem thereby enabling the use of any of the aforementioned algorithms to solve the DSST problem.

Compared with the SSR problem, the DSST problem is much more complicated since there is no longer a central server thus implying that each node only has partial information obtained via its own sensors and messages exchanged with neighboring nodes. Although of significant importance, the current understanding of this more challenging DSST problem is limited. To the best of the authors’ knowledge, only three papers addressed the DSST problem. Moreover, they only solve a subset of the DSST problem since they make assumptions either on system dynamics, network, or both. For example, both~\cite{Distributed1} and \cite{mitra2018distributed} make assumptions on system dynamics. In particular, they require the existence of a basis such that the unobservable space of all the (or enough) sensors is the span of a subset of this basis. We will refer this property as Scalar Decomposability, or SD in brief. Intrinsically, SD enables one to decompose a DSST problem into multiple sub-problems each associated with a scalar system. Interestingly, \cite{Distributed1} and \cite{mitra2018distributed} exploit SD in different ways which leads to different types of solutions of the DSST problem. In~\cite{Distributed1}, with the help of SD, the DSST problem is formulated as a distributed convex optimization problem with time-varying loss function and
a high-gain observer is proposed to reconstruct the state with the help of “blended dynamics approach” introduced in~\cite{lee2020tool}.
In contrast, \cite{mitra2018distributed} proposes a local filter which forces the estimate of an attack-free sensor to always lie in the convex hull of the estimates of its attack-free neighbors. Although this strategy allows extensions to defend against more powerful classes of attacks, it also places an additional assumption on the attacker capabilities by requiring that it does not attack too many neighbors of each network node. Moreover, the tracking algorithm proposed in~\cite{Distributed1} has steady-state error. Compared with~\cite{mitra2018distributed} and~\cite{Distributed1}, the state observer proposed in~\cite{Distributed2} does not require SD, but it is still based on an assumption involving both system dynamics and network topology so it still only solves a subset of DSST problem instances. Moreover, it requires a communication frequency much higher than the sampling rate, which is not typical in applications.

\subsection{Our Approach and Contributions}
In this paper, we study the DSST problem from a new perspective by relating it to the consensus problem. According to this perspective, the objective of each node is to reach consensus on the state of the system. Since the state evolves over time, the relevant type of consensus is the dynamic average consensus. A thorough literature review of the dynamic average consensus problem is provided in the tutorial paper \cite{Tutorial} to which we refer all the interested readers for a discussion of the relevant literature, including~\cite{spanos2005dynamic},~\cite{kia2015dynamic}, and~\cite{Bai}.

Our solution of the DSST problem is based on the simple observation that instead of processing all measurements from all nodes in the network, a suitably compressed version of the measurements suffices for each node to reconstruct the state. In particular, in Section~\ref{sec:DAC} we will show that reconstructing the state from compressed measurements can be formulated as a special case of the SSR problem and, hence, any algorithm that solves the most general case of the SSR problem can be used to solve the DSST problem. Therefore, a solution of the DSST problem is obtained provided each node can track the compressed version of the measurements. To achieve this goal, we draw inspiration from~\cite{Bai} and design an observer that provides each node with error-free tracking of the compressed measurements.

We make the following contributions in this paper:
\begin{enumerate}
    \item We propose a necessary and sufficient condition for the DSST problem to be solvable. 
    \item We provide a solution to the most general case of the DSST problem.
\end{enumerate}

Compared with \cite{Distributed1}, our solution of the DSST problem does not require SD, and the tracking of the state is error-free. Compared with~\cite{Distributed2}, we do no rely on unnecessary assumptions regarding system dynamics or network topology, and we adopt the widely-accepted setting where the communication rate equals the sampling rate. 
The major disadvantage of our algorithm is its computational complexity. Our algorithm is combinatorial in the most general case. Although our algorithm runs in polynomial time when SD holds, it still requires more computations than the algorithms proposed in~\cite{mitra2018distributed} and~\cite{Distributed1}. Moreover, our solution of the DSST problem requires sufficiently fast sampling rates. We summarize the comparison\footnote{Note that the computational complexity is measured by the number of additions and multiplications needed by the algorithm during one round of state update. Although the results in~\cite{Distributed1} are based on a continuous-time representation of the system, computational complexity is assessed for its discrete-time version.} between our solution of the DSST problem and the solutions in~\cite{mitra2018distributed},~\cite{Distributed1}, and~\cite{Distributed2}, in Table~1. 

\begin{table*}[h]
\label{table:intro}
\centering
\begin{tabular}{|y{125pt}||y{115pt}|y{85pt}|}
\hline
 &\cite{mitra2018distributed} & \cite{Distributed1} \\ \hline
Constraints & \tabincell{c}{SD and \\ topology constraints} & SD\\ \hline
Computational complexity & \diagbox{}{} & \diagbox{}{} \\ \hline
\tabincell{c}{Computational complexity \\ if SD holds} & Polynomial & Polynomial \\ \hline
Convergence & Asymptotic & \diagbox{}{} \\ \hline
\tabincell{c}{Existence of \\ steady-state error?} & No & Yes\\ \hline
Working criterion & \tabincell{c}{Continuous and \\ Discrete-time} & Continuous-time \\ \hline
\end{tabular}
\newline
\vspace{0.5 mm}
\newline
\begin{tabular}{|y{125pt}||y{115pt}|y{85pt}|}
\hline
  & \cite{Distributed2} & Our algorithm\\ \hline
Constraints & \tabincell{c}{Constraints on \\ dynamics and topology} & No constraints\\ \hline
Computational complexity & Polynomial but high  & NP-hard \\ \hline
\tabincell{c}{Computational complexity \\ if SD holds} & Not specified & \tabincell{c}{Polynomial but \\ higher than~\cite{Distributed1}} \\ \hline
Convergence  & Asymptotic & Exponential\\ \hline
\tabincell{c}{Existence of \\ steady-state error?} & No & No \\
\hline
Working criterion & \tabincell{c}{High communication \\ frequency} & \tabincell{c}{Requirement on \\ sampling rate} \\
\hline
\end{tabular}
\hspace{1mm}
\caption{A comparison between solutions to the DSST problem in~\cite{mitra2018distributed}, \cite{Distributed1}, \cite{Distributed2}, and this paper.}
\end{table*}

A preliminary version of this work \cite{yanwenconf} was accepted for publication at the 60-th IEEE Conference on Decision and Control.

\subsection{Paper Organization}
The organization of the paper is as follows. We introduce some preliminaries and notations in Section~\ref{sec:notation} which will be used in later sections. The decentralized secure state-tracking problem is formulated in Section~\ref{sec:formulation}, including all the assumptions we made. In Section~\ref{sec:designD} we discuss how to compress the measurements and state a result regarding the solvability of the DSST problem. Then in Section~\ref{sec:reduc} we relate the DSST problem to the dynamic average consensus problem. We provide our solution to the DSST problem in Section~\ref{sec:DAC}. The paper concludes with Section~\ref{sec:conclusion}.

\section{Preliminaries and Notations}
\label{sec:notation}
In this section we introduce the notions used throughout the paper.

\subsection{Basic Notions}
We denote by $\vert S\vert$ the cardinality of a set $S$. For any two sets $S$ and $S'$, the set subtraction $S\backslash S'$ is the set defined by $S\backslash S'=\{s\in S|s\notin S'\}$.

Let $\mathbb{R},~\mathbb{N}$, and $\mathbb{C}$ denote the set of real, natural, and complex numbers, respectively. A complex number $z\in\mathbb{C}$ is written in the form $z=a+b\mathbf{i}$ where $\mathbf{i}^2=-1$. The support of $v\in\mathbb{R}^p$, denoted by $\text{supp}(v)$, is the set of indices of the non-zero entries of $v$, i.e., $\text{supp}(v)=\{i\in\{1,2,\dots,p\}|v_i\ne 0\}$. For a scalar $s\in\mathbb{N}$ we say $v$ is $s-$sparse if $|\text{supp}(v)|\leq s$. Also, we define the all-ones vector $\mathbf{1}_n=(1,1,\hdots,1)^T$ and $I_n$ to be the identity matrix of order $n$. When the dimension of the identity matrix is unambiguous, we will write $I$ instead. Moreover, for any $p\in\mathbb{N}$, we denote by $\mathbf{E}_p=\{\mathbf{e}_1,\mathbf{e}_2,\dots,\mathbf{e}_p\}$ the standard basis of $\mathbb{R}^p$ with $\mathbf{e}_i\in\mathbb{R}^p$ being the vector with all entries equal to zero except the $i$-th entry which is 1.

\subsection{Matrix Related Notions}
We denote by $A\in\mathbb{R}^{m\times n}$ a real matrix with $m\in\mathbb{N}$ rows and $n\in\mathbb{N}$ columns. We will also refer to matrices where only the number of rows or columns is specified using the notation $A\in \mathbb{R}^{m\times *}$ or $A\in \mathbb{R}^{*\times n}$.

For a real square matrix $A$, we denote by $\lambda_{\text{max}}(A)$ and $\lambda_{\text{min}}(A)$ the eigenvalue with the largest and smallest magnitude, respectively. Also, for a real matrix $B$, we denote by $\sigma_{\text{max}}(B)$ and $\sigma_{\text{min}}(B)$ the largest and smallest singular values of matrix $B$, respectively. The range of a real matrix $A$ is denoted by $\mathcal{R}(A)$, and its kernel is denoted by $\mathrm{ker}(A)$. Moreover, we denote by $A\otimes B$ the Kronecker product of two real matrices $A$ and $B$. Consider a set $Q$ of indices and a matrix $K$, the matrix $K_Q$ is obtained by removing any row in $K$ not indexed by $Q$.

\subsection{Graph Related Notions}
Here we review some of the basic notions of graph theory. A weighted undirected graph $\mathcal{G}=(\mathcal{V},\mathcal{E},\mathbf{A})$ is a triple consisting of a set of vertices $\mathcal{V}=\{v_1,v_2,\dots,v_p\}$ with cardinality $p$, a set of edges $\mathcal{E}\subseteq \mathcal{V}\times \mathcal{V}$, and a weighted adjacency matrix $\mathbf{A}\in \mathbb{R}^{p\times p}$ which we will define in the coming paragraph. The set of neighbors of a vertex $i\in\mathcal{V}$, denoted by $\mathcal{N}_i=\{j\in\mathcal{V}|(i,j)\in\mathcal{E}\}$, is the set of vertices that is connected to $i$ by an edge. To clarify, we assume each vertex is not a neighbor of itself, i.e., $(i,i)\notin \mathcal{E}$ for any $i$. The weighted adjacency matrix $\mathbf{A}$ of the graph $\mathcal{G}$ is defined entry-wise. The entry in the $i$-th row and $j$-th column, $a_{ij}$, satisfies $a_{ij}>0$ if $(i,j)\in\mathcal{E}$ and otherwise $a_{ij}=0$. Since the graph is undirected, $a_{ij}=a_{ji}$ for any $i,j$ ranging from 1 to $p$ which results in $\mathbf{A}$ being a symmetric matrix. The degree matrix $\mathbf{D}\in\mathbb{R}^{p\times p}$ of the graph $\mathcal{G}$ is a diagonal matrix with its $i$-th diagonal element defined by $d_{ii}=\sum_{j=1}^pa_{ij}$. The Laplacian matrix $\mathcal{L}$ of the graph $\mathcal{G}$ is defined by $\mathcal{L}=\mathbf{D}-\mathbf{A}$, which is known to be symmetric, if the graph is undirected, positive semi-definite, and having $\text{span}\{\mathbf{1}_p\}$ as its kernel.

\section{Problem Formulation and Key Idea}
\label{sec:formulation}
In this section we introduce the decentralized secure state-reconstruction problem.

\subsection{System Model}
We consider a linear time-invariant system monitored by a network of $p$ nodes whose sensors are subject to attacks:
\begin{equation}
\label{eqn:sys}
\begin{split}
 x[t+1] &=Ax[t],\\
   y_i[t] &= C_ix[t] + e_i[t],
\end{split}
\end{equation}
where $x[t]\in \mathbb{R}^n$ is the system state at time $t\in\mathbb{N}$, $y_i[t]\in\mathbb{R}$ is the measurement of node $i$ where $i\in P\triangleq \{1,2,\dots,p\}$, which is assumed to be a scalar, and the matrices $A$, $B$, and $C_i$ have appropriate dimensions.

The vector $e_i[t]\in\mathbb{R}$ models the attack on the sensor at node $i$ (which we will refer to as sensor $i$ for brevity). If sensor $i$ is attacked by an adversary, then $e_i[t]$ can be arbitrary, otherwise, $e_i[t]$ remains zero for any $t$, and $y_i[t]=C_ix[t]$ holds which means node $i$ receives correct measurements from its sensor. We also assume that the adversary is omniscient, i.e., the adversary has knowledge about the system model, the algorithm being executed at each node and, for any time slot $t$, the adversary knows the system state $x[t]$ and the measurements $y_i[t]$ from all nodes. The only assumption we make on the adversary is that it can only attack a fixed set of at most $s\in\mathbb{N}$ sensors. Note that this set of attacked sensors is unknown to any node in the network.

Collecting $n$ consecutive measurements over time, the output of sensor $i$ can be written in a more compact form:
\begin{equation}
Y_i[t]=\mathcal{O}_ix[t]+E_i[t],\quad i=1,\dots,p,
\end{equation}
where $Y_i[t]$ and $E_i[t]$ are obtained by stacking vertically over time the measurements of sensor $i$ and the attack vector, respectively, and the matrix $\mathcal{O}_i$ is the observability matrix of sensor $i$. These three matrices are defined by:
\begin{equation} \notag
    Y_i[t] = \begin{bmatrix}y_i[t] \\ y_i[t+1] \\ \vdots \\  y_i[t+n-1] \end{bmatrix}\in \mathbb{R}^{n},~E_i[t] = \begin{bmatrix}e_i[t] \\ e_i[t+1] \\ \vdots \\  e_i[t+n-1] \end{bmatrix}\in \mathbb{R}^{n},
\end{equation}
\begin{equation*}
    \mathcal{O}_i=\begin{bmatrix}C_i \\ C_iA \\ \vdots \\ C_iA^{n-1} \end{bmatrix}\in\mathbb{R}^{n\times n}. 
\end{equation*}

In a similar way, we stack over nodes the measurements, the observability matrices, and the attack vectors, from which a more concise representation of the linear system is obtained:
\begin{equation}
\label{eqn:sysstack}
\begin{split}
x[t+1]&=Ax[t], \\
Y[t]&=\mathcal{O}x[t]+E[t],
\end{split}
\end{equation}
where $Y[t]$, $\mathcal{O}$, and $E[t]$ are obtained by stacking vertically each $Y_i[t]$, $\mathcal{O}_i$, and $E_i[t]$, respectively, for $i\in\{1,2,\dots,p\}$, i.e.:
\begin{equation} \notag
    Y[t] = \begin{bmatrix}Y_1[t] \\ Y_2[t] \\ \vdots \\  Y_p[t] \end{bmatrix}\in \mathbb{R}^{pn},~\mathcal{O} = \begin{bmatrix}\mathcal{O}_1 \\ \mathcal{O}_2 \\ \vdots \\  \mathcal{O}_p \end{bmatrix}\in \mathbb{R}^{pn\times n},
\end{equation}
\begin{equation*}
    E[t] = \begin{bmatrix}E_1[t] \\ E_2[t] \\ \vdots \\  E_p[t] \end{bmatrix}\in \mathbb{R}^{pn}. 
\end{equation*}

We also note that since the adversary can only attack at most $s$ nodes, and since for an attack-free sensor $i$ we have $E_i[t]=0$ for any $t$, the vector $E[t]$ is sparse at any $t$. In the end, we define the matrices $y[t]=\begin{bmatrix}y_1[t] & y_2[t] & \dots & y_p[t]\end{bmatrix}^T\in\mathbb{R}^{p}$, $C=\begin{bmatrix}C_1^T & C_2^T & \dots & C_p^T\end{bmatrix}^T\in\mathbb{R}^{p\times n}$ and $e[t]=\begin{bmatrix}e_1[t] & e_2[t] & \dots & e_p[t]\end{bmatrix}^T\in\mathbb{R}^{p}$ for future use.

We assume that the communication between nodes in the network can be modeled by an undirected graph. Each node is modeled by a vertex $i\in\mathcal{V}$, and a communication link from node $i$ to node $j$ is modeled by an edge $(i,j)\in \mathcal{E}$ from vertex $i$ to $j$. Since we assume the graph is undirected, $(i,j)\in \mathcal{E}$ implies $(j,i)\in \mathcal{E}$ which shows that node $j$ can also send messages to node $i$.

\begin{remark}
Although the linear system \eqref{eqn:sysstack} is modelled without inputs, we note that, all results in this paper can be conveniently extended to the case when the input is known to every node in the network.
\end{remark}

\subsection{Assumptions}
Here we list all the assumptions we use in this paper. Some of them have already been discussed when introducing the adversary model. 

\begin{assumption}
\label{assum1}
The network can be modeled by a communication graph which is time-invariant, undirected and connected.
\end{assumption}
\begin{assumption}
\label{assum2}
The adversary is only able to attack at most $s$ nodes. The set of attacked nodes remains constant over time.
\end{assumption}
\begin{assumption}
\label{assum3}
The system dynamics are known to all nodes in the network.
\end{assumption}
\begin{assumption}
\label{assum4}
The adversary is only able to change the measurements of the attacked nodes. Each attacked node still executes its algorithm correctly.
\end{assumption}

The Assumptions~\ref{assum1},~\ref{assum2},~\ref{assum3}, and~\ref{assum4} are in line with the assumptions in~\cite{Distributed1} and~\cite{Distributed2} except that we also assume each node knows the $C_i$ matrices of all other nodes throughout the network. We also require the following assumptions:
\begin{assumption}
\label{assum5}
All the measurements $y_i$ are scalars.
\end{assumption}
\begin{assumption}
\label{assum6}
For any unstable\footnote{Note that an unstable eigenvalue of $A$ is an eigenvalue with magnitude greater or equal to 1. Conversely, the magnitude of a stable eigenvalue is strictly less than 1.} eigenvalue $m+n\mathbf{i}$ of $A$ and for any non-zero eigenvalue $\lambda$ of the communication graph laplacian $\mathcal{L}$, the inequality $\left(m-\frac{\lambda^2}{\lambda^2_{\max}(\mathcal{L})}\right)^2+n^2<1$ holds. 
\end{assumption}

Assumption~\ref{assum5} is not necessary but we adopt it in the context of this paper for brevity. The authors believe that, with slight modifications, any result or algorithm proposed in this paper can be extended to the case where Assumption~\ref{assum5} is dropped.

Whenever the discrete-time linear system \eqref{eqn:sys} is the time discretization of an underlying continuous-time linear system:
\begin{equation}
\label{eqn:sampled}
    \dot{\tilde{x}}=\tilde{A}\tilde{x},
\end{equation}
Assumption \ref{assum6} can be interpreted as a requirement on the sampling time $\tau$ used to obtain \eqref{eqn:sys} from \eqref{eqn:sampled}. If $\tau$ is small enough, then $A = e^{\tilde A\tau}$ can be made arbitrarily close to the identity matrix. In other words, by increasing the sampling rate, $m$ can be made arbitrary close to $1$ and $n$ close to $0$, and such a pair of $m$ and $n$ satisfies Assumption \ref{assum6}.

We also note that, Assumption~\ref{assum6} implies for any eigenvalue $m+n\mathbf{i}$ of $A$, the inequality $(m-1)^2+n^2<1$ holds.

\subsection{The Decentralized Secure State-Tracking Problem} 
In this section we provide the definition of the decentralized secure state-tracking problem.

In plain words, to solve the DSST problem, each node $i$ must maintain a state estimate $\hat x_i[t]$ which converges asymptotically to the true state $x[t]$. We also refer to this property by saying that $\hat x_i[t]$ tracks $x[t]$. The rigorous definition of the DSST problem is as follows:

\begin{definition}[Decentralized Secure State-Tracking Problem]
Consider a linear system subject to attacks~\eqref{eqn:sys} satisfying Assumptions~\ref{assum2}-\ref{assum6} and a communication network satisfying Assumption~\ref{assum1}. The decentralized secure state-tracking problem asks for an algorithm running at each node $i\in P$ with measurements $y_i\in\mathbb{R}$ and messages from neighboring nodes as its input, and such that its output $\hat x_i[t]$ satisfies:  $$\lim_{t\rightarrow \infty}\Vert\hat x_i[t]-x[t]\Vert=0.$$
\end{definition}

\begin{remark}
Differently from centralized SSR problem (see for example~\cite{mao19} for its definition), the DSST problem does not require each node to explicitly know which subset of nodes in the network is attacked. For the SSR problem, it has been argued in~\cite{Hamza14} that knowing the true state $x[t]$ and knowing the set of attacked sensors is equivalent, while this is not the case in the DSST problem setting.
\end{remark}

\begin{remark}
In the DSST problem setting, it is possible to require \textbf{all} nodes, including attacked nodes and attack-free nodes, to maintain a state estimate which asymptotically tracks the system state $x[t]$. This follows from Assumption~\ref{assum1} which restricts the adversary to only alter sensor measurements. In particular, nodes with spoofed measurements are still able to correctly execute their algorithms. Therefore, the attacked nodes are still able reconstruct the state $x[t]$ and may even determine that their own measurements have been altered.
\end{remark}

\subsection{Key Idea}

The key idea for solving the DSST problem is based on the simple observation that instead of having access to measurements $Y=\begin{bmatrix}Y_1^T&Y_2^T&\dots&Y_p^T\end{bmatrix}^T$ of all the sensors, a compressed version $(D\otimes I_n)Y$ of the measurements may suffice to reconstruct the state, where the compression matrix $D\in\mathbb{R}^{v\times p}$ reduces the measurements from $\mathbb{R}^p$ to $\mathbb{R}^v$ with $v\le p$. Compression is possible, in most cases, and thus $v$ will be strictly smaller than $p$. We will elaborate on the feasible choices for a compression matrix $D$ in Section~\ref{sec:designD}.

Equipped with the observation that the compressed version of measurements $(D\otimes I_n)Y$ suffices to reconstruct state, it is natural to ask: how can each node have access to the compressed measurements? We show how to reformulate this problem as a dynamic average consensus problem in Section~\ref{sec:reduc} and an algorithm for each node to track $(D\otimes I_n)Y$ is provided in Section~\ref{sec:DAC}.

Lastly, in Section~\ref{sec:DAC} we show how to reconstruct the state $x$ from the compressed measurements $(D\otimes I_n)Y$ at each node. We will prove that by suitably introducing a slack variable, this problem can be reduced to an SSR problem. This observation implies that we may employ any algorithm for the SSR problem that does not require additional assumptions, even though the compression slightly changes the attack model. It then follows that, as each sensor's estimate of the compressed measurements converges, so does the state reconstructed from the estimated compressed measurements. These three steps provide a solution to the DSST problem.

\section{Design of the Compression Matrix and Solvability of DSST}
\label{sec:designD}

There are two considerations involved in the choice of the compression matrix $D$. On the one hand, in order to reduce communications and storage, we want the $D$ matrix to have the least possible number of rows. On the other hand, the compressed measurements $(D\otimes I_n)Y$ must provide enough information for each node to correctly reconstruct the state. We start with the definition of sparse detectability with respect to a matrix, which is a generalization of sparse observability~\cite{shoukry2015event,chong2015observability} as well as sparse detectability~\cite{nakahira2015dynamic} but stronger, as we will very soon see.

\begin{definition}[Detectability\cite{antsaklis2006linear}]
A pair $(A,C)$ is detectable if all the unobservable eigenvalues of $A$ are stable.
\end{definition}

In other words, if $(A,C)$ is detectable, then for any two trajectories $x_1[t]=A^tx_1[0]$ and $x_2[t]=A^tx_2[0]$, equality $C(x_1[t]-x_2[t])=0$ holding for all $t\in\mathbb{N}$ implies $\lim_{t\rightarrow \infty}(x_1[t]-x_2[t])=0$. This property will be used in the proof of Theorem~\ref{thm:pickd}.

\begin{definition}[Sparse detectability\cite{nakahira2015dynamic}]
The sparse detectability index of the system \eqref{eqn:sys} is the largest integer $k$ such that for any $\mathcal{K}\subseteq P$ satisfying $|\mathcal{K}|\geq p-k$ the pair $(A,C_{\mathcal{K}})$ is detectable. When the sparse detectability index is $k$, we say that system \eqref{eqn:sys} is $k$-sparse detectable.
\end{definition}

In plain words, if we remove any subset of at most $s$ sensors and the remaining system represented by the pair $(A,C_{\mathcal{K}})$ is still detectable, then we say that the original system is $s$-sparse detectable. We now present another perspective on sparse observability, paving the way for the definition of sparse detectability with respect to a matrix, which plays a critical role in our study. We first define the set: $$\mathbf{Q}_s=\{L\in\mathbb{R}^{*\times p}\vert\mathrm{ker}(L)=\text{span}~V,V\subseteq \mathbf{E}_p,|V|\leq s\},$$
that we use in the following equivalent definition of sparse detectability:

\begin{definition}[Sparse detectability]
The sparse detectability index of the system \eqref{eqn:sys} is the largest integer $k$ such that the pair $(A,LC)$ is detectable for any $L\in\mathbf{Q}_k$. When the sparse detectability index is $k$, we say that system \eqref{eqn:sys} is $k$-sparse detectable.
\end{definition}

Intuitively, by left multiplying $C$ by $L$ we remove the measurements of a subset of $s$ sensors without losing information from any remaining sensors. The equivalence between these two definitions of sparse detectability is trivial and, due to space limitations, we do not provide a proof. We proceed by introducing the new notion of sparse detectability with respect to a matrix:

\begin{definition}[Sparse detectability with respect to a matrix]
Consider the system \eqref{eqn:sys}, a matrix $D\in\mathbb{R}^{v\times p}$, and define the following set: $$\mathbf{P}_s=\{L\in\mathbb{R}^{*\times v}\vert\mathrm{ker}(L)=D(\text{\rm{span}}~V),V\subseteq \mathbf{E}_p,|V|\leq s\}.$$ The sparse detectability index of the system \eqref{eqn:sys} with respect to $D$ is the largest integer $k$ such that the pair $(A,LDC)$ is detectable for any $L\in\mathbf{P}_k$. When the sparse detectability index with respect to $D$ is $k$, we say that system \eqref{eqn:sys} is $k$-sparse detectable with respect to $D$.
\end{definition}

\begin{remark}
\label{remark:sparseob}
We note that, sparse detectability with respect to $I$ coincides with the definition of sparse detectability.
\end{remark}
For later use, we formalize the contrapositive of Remark \ref{remark:sparseob} in the next lemma.

\begin{lemma}
\label{lemma:rough}
Any pair $(A,C)$ that fails to be $s$-sparse detectable is not $s-$sparse detectable with respect to any matrix of compatible dimensions.
\end{lemma}

Equipped with the definition of sparse observability with respect to a matrix, we are ready to present necessary conditions for the compressed measurements $(D\otimes I_n)Y$ to carry enough information to reconstruct the state and thus solves the DSST problem.
\begin{lemma}
\label{thm:pickd}

If there exists an algorithm producing a state estimate $\hat{x}$ satisfying $\lim_{t\to\infty} \Vert x[t]-\hat{x}[t]\Vert=0$ for trajectory $x[t]$ of system \eqref{eqn:sys} satisfying Assumptions \ref{assum2}-\ref{assum4} and starting at any initial condition in $\mathbb{R}^n$ only using $(D\otimes I_n)Y[t]$ as input and for any attack signal $e[t]\in \mathbb{R}^p$, then $(A,C)$ is $2s-$sparse detectable with respect to $D$.
\end{lemma}

\begin{proof}
For the sake of contradiction, we assume that $(A,C)$ is not $2s-$sparse detectable with respect to $D$. By the definition of sparse detectability with respect to a matrix, there exist two state trajectories $\tilde x_1[t]=A^t\tilde x_1[0]$ and $\tilde x_2[t]=A^t\tilde x_2[0]$ such that $\lim_{t\rightarrow \infty} (\tilde x_1[t]-\tilde x_2[t])\neq 0$ and
$LDC(\tilde x_1[t]-\tilde x_2[t])=0$ for some $L\in\mathbf{P}_{2s}$ and any $t\in\mathbb{N}$. This also implies the existence of two $s-$sparse signals $e_1[t]\in\mathbb{R}^p$ and $e_2[t]\in\mathbb{R}^p$ such that the equation
$D(C\tilde x_1[t]-C\tilde x_2[t]+e_1[t]-e_2[t])=0$ holds for any $t\in\mathbb{N}$, which is equivalent to:
\begin{equation}
\label{eqn:confuse}
u[t]\triangleq D(C\tilde x_1[t] + e_1[t]) = D(C\tilde x_2[t] + e_2[t]).
\end{equation}
Consider the following 2 scenarios: (1) the state trajectory is $\tilde x_1[t]$ and the attack signal is $e_1[t]$; (2) the state trajectory is $\tilde x_2[t]$ and the attack signal is $e_2[t]$. By~\eqref{eqn:confuse} we know that if a node only has access to $u[t]$ then it cannot distinguish between these 2 scenarios, and since $\lim_{t\rightarrow \infty}(\tilde x_1[t]-\tilde x_2[t])\ne 0$, it is thus unable to track the state.
\end{proof}

Lemma \ref{thm:pickd} shows that it is necessary to pick the compression matrix $D$ such that the pair $(A,C)$ is $2s$-sparse detectable with respect to $D$, otherwise it is impossible to track the state from compressed measurements $(D\otimes I_n)Y$. Later in Lemma \ref{thm:mainthm} we will see that this condition also suffices to solve the DSST problem. Anticipating the results in Section \ref{sec:DAC} we state here the main contribution of this paper, $2s$-sparse detectability with respect to $D$ is both necessary and sufficient to solve the DSST problem:

\begin{theorem}
\label{thm:solvability}
The DSST problem associated with the linear system subject to attacks defined in~\eqref{eqn:sys} satisfying Assumptions~\ref{assum2}-\ref{assum6} and a communication network satisfying Assumption~\ref{assum1} is solvable if and only if the pair $(A,C)$ is 2s-sparse detectable.
\end{theorem}
\begin{proof}
Necessity comes easily from Lemma \ref{thm:pickd} by picking $D=I_p$. Sufficiency follows from Lemma \ref{thm:mainthm} in Section \ref{sec:DAC}.
\end{proof}

\section{Reduction to Dynamic Average Consensus}
\label{sec:reduc}
One key step of our solution to the DSST problem is to have each node tracking the compressed measurements $(D\otimes I_n)Y$. To do this, we ask each node $i$ to maintain an estimate vector $W_i=\begin{bmatrix}(W_i^1)^T & (W_i^2)^T & \dots & (W_i^v)^T\end{bmatrix}^T\in\mathbb{R}^{vn}$, whose $j$-th block, $W_i^j$, tracks the $j$-th linear combination of measurements $\sum_i d_{ji}Y_i$, where $d_{ji}$ is the entry at the $j$-th row and $i$-th column of $D$.

In this section, and in the one that follows, we focus on the problem of tracking $(D_1\otimes I_n)Y$, where $D_1$ is the first row of $D$, which can also be written as $\sum_id_{1i}Y_i$, and the tracked value is stored in the block $W_i^1$. For technical reasons we will use $\frac{1}{p}\sum_id_{1i}Y_i$, which serves the same purpose since by Assumption~\ref{assum3} the value of $p$ is known to all nodes. Note that any algorithm that can track $\frac{1}{p}\sum_id_{1i}Y_i$ can be extended to track $\frac{1}{p}(D\otimes I_n)Y$ but running $v$ concurrent copies, each with a different set of weights $\{d_{ji}\}$. 
We observe that this problem can be seen as an instance of the dynamic average consensus problem~\cite{Tutorial}. In brief, suppose that each agent in the network has a time-varying local reference signal $\phi_i(t): [0,\infty)\rightarrow \mathbb{R}^{n}$. The dynamic average consensus problem asks for an algorithm that allows individual agents to track the time-varying average of the reference signals, given by:
\begin{equation}
u^{\text{avg}}[t]=\frac{1}{p}\sum_{i=1}^p\phi_i[t].
\end{equation}

In our problem setting, we pick the external input corresponding to $(D_1\otimes I_n)Y$ to be $\phi_i^1[t]=d_{1i}Y_i[t]$ and what we want to track is $\frac{1}{p}\sum_{i=1}^p\phi_i^1[t]$. In other words, we may adopt any algorithm that solves the dynamic average consensus algorithm thereby enabling all nodes to track $(D\otimes I_n)Y$. 

However, in the setting of dynamic average consensus problem, no knowledge about reference signals $\phi_i$ is assumed, whereas in our problem, for an attack-free node $i$, we have the following:
\begin{equation}
\label{eqn:consis}
Y_i[t+1] = \underbrace{\begin{bmatrix} 0 & 1 & 0 & \dots & 0 \\ 0 & 0 & 1 & \dots & 0 \\ \vdots & \vdots & \vdots & \ddots & \vdots \\ -\alpha_0 & -\alpha_1 & -\alpha_2 & \dots & -\alpha_{n-1} \end{bmatrix}}_{\hat A} Y_i[t].
\end{equation}
This equality comes from the construction $Y_i[t]=\begin{bmatrix}y_i[t]& y_i[t+1]&\dots & y_i[t+n-1]\end{bmatrix}^T$ where $\alpha_1,\alpha_2,\dots,\alpha_{n-1}$ are the coefficients of the characteristic polynomial of $A$. This can be seen by nothing that $y_i[t+n]=C_iA^nx[t]=C_i(\alpha_{n-1}A^{n-1}+\dots +\alpha_1A+\alpha_0)x[t]=\alpha_{n-1}y_i[t+n-1]+\dots+\alpha_1 y_i[t+1]+\alpha_0y_i[t]$ whenever node $i$ is attack free. Moreover, we note that $\hat A$ is the controller form of $A$ and has the same eigenvalues as $A$.

\begin{remark}
Since the state portion corresponding to the stable eigenvalues will decay to zero with the elapse of time, in the remainder of the paper we will only be focusing on the state portion corresponding to unstable eigenvalues.
\end{remark}

Based on Equation~\eqref{eqn:consis}, we define the following local sanity check:

\noindent\rule{12.62cm}{0.4pt}

\noindent \textbf{Local sanity check}: given measurement $y_i[t]$ of node $i$, we say node $i$ passes the local sanity check if $\Vert Y_i[t+1]-\hat A Y_i[t]\Vert\leq \epsilon$ for any $t>0$, where $\epsilon$ is a chosen error tolerance.

\noindent\rule{12.62cm}{0.4pt}

In practice, we ask all nodes in the network to constantly run the local sanity check. It is trivially seen that all attack-free nodes would pass the check at any time. On the other hand, if a node fails the local sanity check, then we immediately reach the conclusion that this node is under attack. Therefore, we assume that all the attack vectors $e_i$ corresponding to attacked nodes are constructed so that the resulting measurements $Y_i[t]$ pass the local sanity check, i.e., $\Vert Y_i[t+1]-\hat A Y_i[t]\Vert\leq \epsilon$. 

The purpose of the local sanity check is to force the dynamics of $Y_i$ to be governed by \eqref{eqn:consis}, including those from attacked nodes. We will exploit this additional knowledge of $d_{1i}Y_i$ (or $\phi_i^1$) to
achieve better tracking results.

\section{Solving the DSST Problem}
\label{sec:DAC}
We argued in Section~\ref{sec:reduc} that the tracking of the compressed measurements $(D\otimes I_n)Y$ is intrinsically identical to a dynamic average consensus problem. We extend the results in~\cite{Bai} from the scalar-input case to the vector case and present such extension in the discrete-time domain. 

\subsection{Tracking the Compressed Measurements}

We ask each node to update its estimate of $\frac{1}{p}\sum_id_{1i}Y_i[t]$ according to:

\begin{equation}
\label{eqn:localestimate}
\left\{
\begin{aligned}
W_i^1[t+1] & = && (\hat A-I)W_i^1[t] - 2k_I\sum_{j\in\mathcal{N}_i}(\eta_j[t]-\eta_i[t]) + \phi_i^1[t], \\
b_i[t+1] & = && \hat Ab_i[t] + k_I\sum_{j\in \mathcal{N}_i}(W_j^1[t]-W_i^1[t]), \\
\eta_i[t] & = && k_Pb_i[t] + k_I\sum_{j\in \mathcal{N}_i}(W_j^1[t]-W_i^1[t]),
\end{aligned}
\right.
\end{equation}
where $W_i^1[t]$ is the estimate at node $i$ of the average of the input signal $\frac{1}{p}\sum_i \phi_i^1[t]$, $k_I,k_P\in\mathbb{R}$ are design parameters, $b_i[t]$ and $\eta_i[t]$ are internal states of node $i$. The following theorem states that under suitable choices of $k_I$ and $k_P$ the estimate of each node $W_i^1[t]$ approaches the true state $\frac{1}{p}\sum_i \phi_i^1[t]$.

\begin{theorem}
\label{thm:main}
Consider the average tracking algorithm in~\eqref{eqn:localestimate} where the input signal satisfies $\phi_i^1[t+1]=\hat A\phi_i^1[t]$ for all $i\in\{1,2,\dots,p\}$. There exist constants $k_{I}\in\mathbb{R}$ and $k_{P}\in\mathbb{R}$ such that the state estimate $W_i^1[t]$ at each node tracks the average of the input signals exponentially fast, i.e.: $$\left\Vert W_i^1[t]-\frac{1}{p}\sum_i \phi_i^1[t] \right\Vert <\beta \alpha^t,$$ for some $0<\alpha<1$ and $\beta>0$.
\end{theorem}

Before proving this theorem we present the following two lemmas that will be used in its proof.

\begin{lemma}
\label{lemma:diagonalize}
Consider a positive semi-definite matrix $B=B^T\in\mathbb{R}^{n\times n}$ and a matrix $S\in\mathbb{R}^{n\times m}$ such that $\mathcal{R}(S)\cap \mathrm{ker}(B)=\{0\}$. The matrix $S^TBS$ is diagonalizable and positive definite.
\end{lemma}

\begin{proof}
The matrix $S^TBS$ being symmetric implies it is diagonalizable and positive semi-definite since B is also positive semi-definite. Assume, for the sake of contradiction, that it is not positive definite. Then, there exists $x\neq 0\in \mathbb{R}^n$ such that $x^TS^TBSx=0$. Since $x^TS^TBSx=\Vert \sqrt{B}Sx\Vert_2^2$ we conclude that $\sqrt{B}Sx=0$, i.e., $Sx\in \mathrm{ker}(B)$ since any matrix $B$ and its square root $\sqrt{B}$ have the same kernel, a contradiction with $\mathcal{R}(S)\cap \mathrm{ker}(B)=\{0\}$.
\end{proof}

\begin{lemma}
\label{lemma:hurwitzset}
Consider a block matrix $B=\begin{bmatrix} A-I-2k_I^2\lambda^2I_n & -2k_Ik_P\lambda I_n \\ k_I\lambda I_n & A \end{bmatrix}$ where $A\in\mathbb{R}^{n\times n}$ and $\lambda$ is a non-zero eigenvalue of $\mathcal{L}$. There always exist $k_I,k_P\in\mathbb{R}$ such that $B$ is stable\footnote{We recall that a matrix is stable if and only if the magnitude of all its eigenvalues is strictly less than 1.} if $A$ satisfies Assumption~\ref{assum6} in Section~\ref{sec:formulation}.
\end{lemma}

To prove this result, we first invoke the following lemma:

\begin{lemma}
\label{lemma:onion}
Consider a block square matrix $H=\begin{bmatrix}
A & B \\ C & D
\end{bmatrix}$, where $A,B,C,D\in\mathbb{R}^{r\times r}$ and $r\in\mathbb{N}$ satisfy the following properties: (1) blocks $A$ and $D$ are upper-triangular and have $a_1,a_2,\dots,a_r$ and $d_1,d_2,\dots,d_r$ as their diagonal entries respectively, and (2) blocks $B$ and $C$ are diagonal and have $b_1,b_2,\dots,b_r$ and $c_1,c_2,\dots,c_r$ as their diagonal entries respectively. The matrix $H$ is stable if every $2\times 2$ matrix $\begin{bmatrix}
a_i & b_i \\ c_i & d_i
\end{bmatrix}$ is stable for $i$ ranging from 1 to $r$.
\end{lemma}

\begin{proof}
We will only prove the special case when $r=2$ since the rest of the proof follows by an induction argument, which we omit for the sake of brevity.

We explicitly write out the matrix $H$: $$H=\left[
\begin{array}{cc|cc} 
    a_1 & a_0 & b_1 & \\ 
        & a_2 &     & b_2 \\ \hline
    c_1 &     & d_1 & d_0 \\ 
        & c_2 &     & d_2 \\
\end{array}
\right],$$
and compare with another matrix $\hat H$:
$$\hat H=\left[
\begin{array}{cc|cc} 
    a_1 & b_1 & a_0 & \\ 
    c_1 & d_1 &     & d_0 \\ \hline
        &     & a_2 & b_2 \\ 
        &     & c_2 & d_2 \\
\end{array}
\right].$$ We observe that $\hat H$ can be obtained from $H$ by a similarity transformation: $\hat H=T^{-1}HT$, where $$T = T^{-1}= \begin{bmatrix}1&0&0&0\\0&0&1&0\\0&1&0&0\\0&0&0&1\end{bmatrix}.$$ In other words, matrices $H$ and $\hat H$ share the same set of eigenvalues. Moreover, we obtain that $\hat H$ is a block diagonal matrix with both of its diagonal blocks $\begin{bmatrix}
a_1 & b_1 \\ c_1 & d_1
\end{bmatrix}$ and $\begin{bmatrix}
a_2 & b_2 \\ c_2 & d_2
\end{bmatrix}$ being stable. This shows $\hat H$ is a stable matrix, which finishes our proof.

\end{proof}

With the help of Lemma~\eqref{lemma:onion} we are ready to prove that $B$ can be made stable by properly choosing $c$ and $\lambda$.

\vspace{1.5mm}
\noindent\textit{Proof of Lemma \ref{lemma:hurwitzset}.}
We first find the Jordan decomposition of $A$. Let $J=T^{-1}AT$ be the Jordan form of $A$, $J$ is an upper-triangular matrix with entries $a_1,a_2,\dots,a_n$ on its diagonal. Instead of matrix $B$, we focus on the following matrix which has the same eigenvalues as $B$:

\begin{eqnarray*}
&& \begin{bmatrix}T^{-1}& \\ & T^{-1}\end{bmatrix} \begin{bmatrix} A-I-2k_I^2\lambda^2I_n & -2k_Pk_I\lambda I_n \\ k_I\lambda I_n & A \end{bmatrix} \begin{bmatrix}T& \\ & T\end{bmatrix} \\
&=& \begin{bmatrix} J-I-2k_P^2\lambda^2I_n & -2k_Pk_I\lambda I_n \\ k_P\lambda I_n & J \end{bmatrix}.
\end{eqnarray*}

Now we invoke Lemma~\eqref{lemma:onion} and obtain that $B$ is stable if the matrix $R(a,\lambda)=\begin{bmatrix}a-1-2k_I^2\lambda^2 & -2k_Pk_I\lambda \\ k_I\lambda & a\end{bmatrix}$ is stable for any $a\in\{a_1,a_2,\dots,a_n\}$ being an eigenvalue of $A$ and any $\lambda$ being a non-zero eigenvalue of $\mathcal{L}$. In the next step we compute the eigenvalues $k_1,k_2$ of matrix $R(a,\lambda)$:
\begin{eqnarray} \notag
&&(a-1-2k_I^2\lambda^2-k)(a-k)+2k_Pk_I^2\lambda^2=0 \\ \notag
\Rightarrow && k^2+(2k_I^2\lambda^2-2a+1)k+(2k_pk_I^2\lambda^2+a^2-a-2ak_I^2\lambda^2) =0.
\end{eqnarray}
The eigenvalues of the matrix $R(a,\lambda)$ are given by:
\begin{eqnarray*}
2k_1 &=& -2k_I^2\lambda^2+2a-1 + \sqrt{(2k_I^2\lambda^2+1)^2-8k_Pk_I^2\lambda^2}, \\
2k_2 &=& -2k_I^2\lambda^2+2a-1 - \sqrt{(2k_I^2\lambda^2+1)^2-8k_Pk_I^2\lambda^2}.
\end{eqnarray*}

Lastly, we provide a possible choice of $\lambda$ and $c$ such that for any aforementioned $a$ and $\lambda$, both eigenvalues $k_1$ and $k_2$ are stable. We pick $k_P=1$ and $k_I=\frac{1}{\sqrt{2}\lambda_{\max}(\mathcal{L})}$, where we recall that $\lambda_{\max}(\mathcal{L})$ is the largest eigenvalue of $\mathcal{L}$. This choice leads to: $$k_1=a-1,~k_2=a-\frac{\lambda^2}{\lambda_{\max}^2(\mathcal{L})}.$$ By Assumption~\ref{assum6} both $k_1$ and $k_2$ are stable\footnote{Again, we recall that a eigenvalue is stable if its magnitude is strictly less than 1.} for any feasible choice of $a$ and $\lambda$, which finishes our proof.\hfill\qedsymbol
\vspace{1.5mm}

Instead of proving directly Theorem~\ref{thm:main}, for the sake of simplicity, we prove the convergence result for a broader class of systems in the form of:
\begin{equation}
\label{eqn:genestimate}
\left\{
\begin{aligned}
v[t+1] & = && (G-I)v[t] + \varphi[t] - 2k_IF\eta[t], \\
m[t+1] & = && Gm[t] + k_IFv[t],\\
\eta[t] & = && k_Pm[t]+ k_IFv[t],
\end{aligned}
\right.
\end{equation}
where $k_I,k_P\in \mathbb{R}$, $v,m,\varphi,\eta\in\mathbb{R}^a$ for some positive $a\in\mathbb{N}$, $F,G\in\mathbb{R}^{a\times a}$, $F=F^T$ is positive semi-definite, and $I$ represents the identity matrix of order $a$. 

We first show that tracking algorithms of the form of~\eqref{eqn:genestimate} can be decoupled into 2 sub-systems and then we analyze the convergence properties for each sub-system. To do this, we introduce a matrix $R\in\mathbb{R}^{a\times \mathrm{dim}(\mathrm{ker}F)}$ such that $\mathcal{R}(R)=\mathrm{ker}(F)$. Moreover, we pick $S\in\mathbb{R}^{a\times (a-\mathrm{dim}(\mathrm{ker}F))}$ such that the matrix $\begin{bmatrix}R & S\end{bmatrix}\in\mathbb{R}^{a\times a}$ is orthogonal.

We make the following three additional assumptions regarding matrices $F$ and $G$: 
\begin{itemize}
\item \textbf{Assumption 1:} The external input $\varphi[t]$ satisfies $\varphi[t+1]=G\varphi[t]$, where $G$ has the following two properties:
\begin{enumerate}
    \item All eigenvalues of $G$ satisfy Assumption~\ref{assum6} in Section~\ref{sec:formulation}.
    \item The equations $R^TGS=0$ and $S^TGR=0$ hold.
\end{enumerate}
\item \textbf{Assumption 2:}  All eigenvalues of $F$ coincide with the eigenvalues of $\mathcal{L}$.
\item \textbf{Assumption 3:} The matrices $F$ and $G$ are simultaneously block diagonalizable and each block in $G$ corresponds to an identity block (possibly scaled by a scalar $\lambda_i$) in $F$. In other words, there exists an invertible matrix $P$ such that $P^{-1}GP=\text{diag}\{\Lambda_1,\dots,\Lambda_r\}$ and $P^{-1}FP=\text{diag}\{\lambda_1I_{s_1},\dots,\lambda_rI_{s_r}\}$ such that $\Lambda_i\in\mathbb{R}^{s_i\times s_i}$ with $i$ ranging from $1$ to $r$.
\end{itemize}

Note that all the additional assumptions are introduced to prove the desired tracking property and any system of the form \eqref{eqn:localestimate} automatically falls into the broader class of systems described in~\eqref{eqn:genestimate} and the additional Assumptions 1-3, as we will argue at the end of this section. We are now ready to present our decomposition result based on the change of coordinates $(z_1,z_2)=(R^T,S^T)v$ where $z_1$ describes the sum of local estimates and $z_2$ is the difference among local estimates. We see that if $z_1$ converges to $R^T\varphi$ and $z_2$ converges to $0$, then all nodes reach the consensus $\varphi$ which is the target of the tracking algorithm.

\begin{proposition}
\label{prop:decom}
The dynamical system in~\eqref{eqn:genestimate} can be decoupled into the following two sub-systems if it satisfies additional Assumption 1:
\begin{equation}
\label{eqn:direction1}
z_1[t+1]=R^T(G-I)Rz_1[t]+R^T\varphi[t],
\end{equation}
and:
\begin{equation}
\label{eqn:direction2}
\left\{
\begin{aligned}
z_2[t+1] & = && S^T(G-I)Sz_2[t] - 2k_IS^TFSf_2[t] + S^T\varphi[t], \\ g_2[t+1] & = && S^TGSg_2[t] + k_IS^TFSz_2[t],\\
 f_2[t] &= && k_Pg_2[t]+ k_IS^TFSz_2[t],
\end{aligned}
\right.
\end{equation}
where $z_1=R^Tv$, $z_2=S^Tv$, $f_2=S^T\eta$ and $g_2=S^Tm$.
\end{proposition}

\begin{proof}
Let $z=\begin{bmatrix}z_1^T & z_2^T\end{bmatrix}^T=\begin{bmatrix}R & S\end{bmatrix}^Tv$, $g=\begin{bmatrix}g_1^T & g_2^T\end{bmatrix}^T=\begin{bmatrix}R & S\end{bmatrix}^Tm$, and $f=\begin{bmatrix}f_1^T & f_2^T\end{bmatrix}^T=\begin{bmatrix}R & S\end{bmatrix}^T\eta$, which in turn means $v=\begin{bmatrix}R & S\end{bmatrix}z$, $m=\begin{bmatrix}R & S\end{bmatrix}g$, and $\eta=\begin{bmatrix}R & S\end{bmatrix}f$. By substituting these changes of coordinates into~\eqref{eqn:genestimate} and also left multiplying $\begin{bmatrix}R & S\end{bmatrix}^T$ on both sides of Equation \eqref{eqn:genestimate} we obtain:
\begin{equation}
\left\{
\begin{aligned}
z[t+1] & = && \begin{bmatrix}R^T \\ S^T \end{bmatrix}(G-I)\begin{bmatrix}R & S\end{bmatrix}z[t] - 2k_I\begin{bmatrix}R^T \\ S^T \end{bmatrix}F\begin{bmatrix}R & S\end{bmatrix}f[t] + \begin{bmatrix}R^T \\ S^T \end{bmatrix}\varphi[t],  \\
g[t+1] & = && \begin{bmatrix}R^T \\ S^T \end{bmatrix}G\begin{bmatrix}R & S\end{bmatrix}g[t] + k_I\begin{bmatrix}R^T \\ S^T \end{bmatrix}F\begin{bmatrix}R & S\end{bmatrix}z[t],\\
f[t] & = && k_P\begin{bmatrix}R^T \\ S^T \end{bmatrix}\begin{bmatrix}R & S\end{bmatrix}g[t]+ k_I\begin{bmatrix}R^T \\ S^T \end{bmatrix}F\begin{bmatrix}R & S\end{bmatrix}z[t].
\end{aligned}
\right.
\end{equation}
Recall that $R^TF=0$, $FR=0$, by additional Assumption 1 we have $R^TGS=0$ and $S^TGR=0$, and by construction $R^TS=0$, after some manipulations we obtain:
\begin{equation}
z_1[t+1]=R^T(G-I)Rz_1[t]+R^T\varphi[t],
\end{equation}
and
\begin{equation}
\left\{
\begin{aligned}
z_2[t+1] & = && S^T(G-I)Sz_2[t] - 2k_IS^TFSf_2[t] + S^T\varphi[t], \\
g_2[t+1] & = && S^TGSg_2[t] + k_IS^TFSz_2[t],\\
 f_2[t] &= && k_Pg_2[t]+ k_IS^TFSz_2[t].
\end{aligned}
\right.
\end{equation}
as desired.
\end{proof}

The following proposition shows that $z_1$ converges to $R^T\varphi$ as desired.

\begin{proposition}
\label{prop:avg}
For any solution $z_1[t]$ of~\eqref{eqn:direction1} we have: $$\left\Vert z_1[t]-R^T\varphi[t]\right\Vert<\beta \alpha^t,$$ for some $0<\alpha<1$ and $\beta>0$.
\end{proposition}

\begin{proof}
Let $\hat z_1[t]=z_1[t]-R^T\varphi[t]$, we have the following sequence of equations:
\begin{eqnarray}
{\hat z}_1[t+1] & = & z_1[t+1] - R^T \varphi[t+1] \\ \notag
&=& R^T(G-I)Rz_1[t] + R^T\varphi[t] - R^TG\varphi[t]  \\ \notag
&=& R^T(G-I)Rz_1[t] + R^T(I-G)\varphi[t] \\ \notag
&=& R^T(G-I)Rz_1[t] + R^T(I-G)RR^T\varphi[t] \\ \notag
&=& R^T(G-I)R \hat z_1[t],
\end{eqnarray}
where the fourth step comes from: 
\begin{equation}
R^TG = R^TG\begin{bmatrix}R & S\end{bmatrix}\begin{bmatrix}R^T \\ S^T\end{bmatrix}=\begin{bmatrix}R^TGR & 0\end{bmatrix}\begin{bmatrix}R^T \\ S^T\end{bmatrix} = R^TGRR^T,
\end{equation}
since $R^TGS=0$ by additional Assumption 1, and by construction of $R$ we have $R^TRR^T=R^T$.

Finally, we show that the all eigenvalues of $R^T(G-I)R$ are stable, in other words, their magnitudes are strictly less than 1. To do this, we first note that if $\lambda'=m-1+n\textbf{i}$ is an eigenvalue of $G-I$, then $\lambda=m+n\textbf{i}$ must be an eigenvalue of of $G$. By additional Assumption 3 and Assumption~\ref{assum6}, for any $\lambda=m+n\textbf{i}$ of $G$, we have $\Vert \lambda'\Vert_2^2=(m-1)^2+n^2<1$, which shows that all eigenvalues $\lambda'$ of $G-I$ are stable.

Then we resort to the 
following similarity transformation:
\begin{eqnarray*}
\begin{bmatrix} R^T \\ S^T \end{bmatrix} (G-I)\begin{bmatrix} R & S \end{bmatrix}&=&\begin{bmatrix} R^T(G-I)R & R^T(G-I)S \\ S^T(G-I)R & S^T(G-I)S \end{bmatrix} \\
&=&\begin{bmatrix}R^T(G-I)R & 0 \\ 0 & S^T(G-I)S\end{bmatrix},
\end{eqnarray*}
which shows that the eigenvalues of $R^T(G-I)R$ are a subset of eigenvalues of $G-I$. Since all the eigenvalues of $G-I$ are stable, all eigenvalues of $R^T(G-I)R$ are also stable and hence $\hat z_1$ converges to 0 exponentially fast, or $z_1$ converges to $R^T\varphi$ exponentially fast.
\end{proof}

Similarly, we prove that $z_2$ converges to $0$ in the following proposition.
\begin{proposition}
\label{prop:cons}
For any solution $(z_2[t],g_2[t])$ of~\eqref{eqn:direction2} we have: $$\left\Vert z_2[t]\right\Vert<\beta \alpha^t,$$ for some $0<\alpha<1$ and $\beta>0$ if the system defined in~\eqref{eqn:genestimate} satisfies additional Assumptions 1 and 2.
\end{proposition}

\begin{proof}
By substituting $f_2$ by $k_Pg_2+ k_IS^TFSz_2$ we obtain the following dynamics:
\begin{equation} \notag 
\label{eqn:consensus}
\begin{bmatrix} z_2[t+1] \\ g_2[t+1] \end{bmatrix}=\underbrace{\begin{bmatrix}G'-I-2k_I^2F'^2 & -2k_Pk_IF' \\ k_IF' & G'\end{bmatrix}}_{B}\begin{bmatrix} z_2[t] \\ g_2[t]\end{bmatrix} 
+ \begin{bmatrix}S^T\\ 0\end{bmatrix}\varphi[t],
\end{equation}
where $G'=S^TGS$ and $F'=S^TFS$. It is simple to see that, if $B$ is stable, for any two solutions $w$ and $w'$ of~\eqref{eqn:consensus}, we have $\lim_{t\to\infty}(w[t]-w'[t])=0$, since $z=w-w'$ satisfies $z[t+1]=Bz[t]$. In other words, if we can prove: $\textit{1}.$ there exists a function $\hat{g}_2[t]:\mathbb{R}\rightarrow \mathbb{R}^{a-\mathrm{dim}(\mathrm{ker}F)}$ such that $(0,\hat{g}_2)$ is a solution of ~\eqref{eqn:consensus}, and $\textit{2}.$ the matrix $B$ is Hurwitz, then the proposition is proved.

First we prove that $(0, \hat g_2[t])$ is the solution of~\eqref{eqn:consensus} for some $\hat g_2[t]$. By substituting $z_2=0$ into~\eqref{eqn:consensus} we obtain:
\begin{equation}
\label{eqn:simplifyeq}
\left\{
\begin{aligned}
S^T\varphi[t] &=&& 2k_Pk_IF'\hat g_2[t],\\
{\hat g}_2[t] &=&& G' \hat g_2[t].
\end{aligned}
\right.
\end{equation}
By Lemma~\ref{lemma:diagonalize} and the construction of $S$ we obtain that $F'$ is an invertible matrix and hence $(\hat z_2=0,\hat g_2[t]=\frac{1}{2k_Pk_I}F'^{-1}S^T\varphi[t])$ is a feasible solution. To prove this, we show that if the first equality in~\eqref{eqn:simplifyeq} holds at time $t$, then it will also hold at time $t+1$. Then a simple induction can be conducted and the proof could thus be obtained. We check that:
\begin{eqnarray}
2k_Pk_IF' \hat g_2[t+1] &\overset{\text{(a)}}{=}& 2k_Pk_IF'G'\hat g_2[t] \\ \notag
&\overset{\text{(b)}}{=}& 2k_Pk_IS^TFSS^TGS\hat g_2[t] \\ \notag
&\overset{\text{(c)}}{=}& 2k_Pk_IS^TFGS\hat g_2[t] \\ \notag
&\overset{\text{(d)}}{=}& 2k_Pk_IS^TGFS\hat g_2[t] \\ \notag
&\overset{\text{(e)}}{=}& 2k_Pk_IS^TGSS^TFS\hat g_2[t] \\ \notag
&\overset{\text{(f)}}{=}& 2k_Pk_IG'F'\hat g_2[t] \\ \notag
&\overset{\text{(g)}}{=}& G'S^T\varphi[t] \\ \notag
&\overset{\text{(h)}}{=}& S^TGSS^T\varphi[t] \\ \notag
&\overset{\text{(i)}}{=}& S^T\varphi[t+1],
\end{eqnarray}
where in steps (c)(e)(h) we repeatedly used the inequalities $S^TGSS^T=S^TG$ and $S^TFSS^T=S^TF$ which could be obtained by following a similar argument as in the proof of Proposition~\ref{prop:avg}, and step (d) is true since by additional Assumption 2 matrices $F$ and $G$ commute. This finishes the proof that $(0,\hat g_2(t))$ is a solution of Equation~\eqref{eqn:consensus}.

We now prove that $B$ can be made stable by suitably choosing $k_I$ and $k_P$. First of all, we show that $G'=S^TGS$ and $F'=S^TFS$ are also simultaneously block diagonalizable and each pair of corresponding blocks has the property described in additional Assumption 2. To see why this is true, consider the invertible matrix $P'=\begin{bmatrix}R^T \\ S^T\end{bmatrix}P$. We observe that $P'$ simultaneously block-diagonalizes $\begin{bmatrix}R^T \\ S^T\end{bmatrix}G\begin{bmatrix}R S\end{bmatrix}$ and $\begin{bmatrix}R^T \\ S^T\end{bmatrix}F\begin{bmatrix}R S\end{bmatrix}$ if $P$ simultaneously block-diagonalizes $F$ and $G$. Moreover, since $R^TGS=0$ and $S^TGR=0$, both $\begin{bmatrix}R^T \\ S^T\end{bmatrix}G\begin{bmatrix}R S\end{bmatrix}$ and $\begin{bmatrix}R^T \\ S^T\end{bmatrix}F\begin{bmatrix}R S\end{bmatrix}$ are block diagonal matrices which yields that $P'$ must also be a block diagonal matrix, with its bottom-right block $P_2$ at the size of $S^TGS$ simultaneously diagonalizing $S^TGS$ and $S^TFS$. 

We then resort to the matrix $P_2$ to continue our proof. We consider the following similarity transformation: $$B'=\begin{bmatrix}B_{11}&B_{12}\\B_{21}&B_{22}\end{bmatrix}=\begin{bmatrix}P_2^{-1}&0\\0&P_2^{-1}\end{bmatrix}B\begin{bmatrix}P_2&0\\0&P_2\end{bmatrix},$$ under which the eigenvalues of $B$ are invariant. Simple computations show that $B_{11}=\text{diag}\{-\Lambda_l-2k_I^2\lambda_l^2I,\dots,-\Lambda_r-2k_I^2\lambda_r^2I\}$, $B_{21}=\text{diag}\{k_I\lambda_lI,\dots,k_I\lambda_rI\}$,
$B_{12}=\text{diag}\{-2k_Pk_I\lambda_lI,\dots,-2k_Pk_I\lambda_rI\}$,  and $B_{22}=\text{diag}\{\Lambda_l,\dots,\Lambda_r\}$ for some $l<r$. The block diagonal structure of each block in $B'$ yields the fact that the eigenvalues of both $B'$ and $B$ are the eigenvalues of the smaller matrices:
\begin{equation}
    \label{eqn:particle}
    \begin{bmatrix}\Lambda_i-I-2k_I^2\lambda_j^2I & -2k_Pk_I\lambda_jI \\ k_I\lambda_jI & \Lambda_i\end{bmatrix},\quad i,j\in\{l,\dots,r\},
\end{equation}
where any eigenvalue of any $\Lambda_i$ is an eigenvalue of $G$, and $\lambda_j>0$ is a non-zero eigenvalue of $F$. Here we invoke Lemma~\ref{lemma:hurwitzset} and conclude that if the matrix $G$ satisfy the same constraint on its eigenvalues in Assumption~\ref{assum6}, and $\lambda$ being an eigenvalue of $S^T\mathcal{L}S$, both of which are true by additional Assumption (2), then there exists a universal choice of $k_I$ and $k_P$ such that for any feasible choice of $\lambda$, the matrix $B$ is stable.

In conclusion, by properly choosing $k_I$ and $k_P$, $z_2[t]$ converges to $0$ exponentially fast.
\end{proof}

Lastly, we go back to prove Theorem~\ref{thm:main}. 
\begin{proof}
The dynamic average tracking algorithm in~\eqref{eqn:localestimate} can be compactly written in the form:

\begin{equation}
\label{eqn:globalestimate}
\left\{
\begin{aligned}
W^1[t+1] & = && (I_p\otimes \hat A-I_{np})W^1[t] + 2(I_p\otimes \hat A)\phi^1[t] \\
&&&  - 2k_I(\mathcal{L}\otimes I_n)\eta[t], \\
b[t+1] & = && (I_p\otimes \hat A)b[t] + k_I(\mathcal{L}\otimes I_n)W^1[t],\\
\eta[t] & = && k_Pb[t]+ k_I(\mathcal{L}\otimes I_n)W^1[t],
\end{aligned}
\right.
\end{equation}
where $W^1 = \begin{bmatrix}(W_1^1)^T & (W_2^1)^T & \dots & (W_p^1)^T\end{bmatrix}^T$ and $\phi^1=\begin{bmatrix}(\phi_1^1)^T & \dots & (\phi_p^1)^T\end{bmatrix}^T$ are obtained by concatenating estimates and weighted measurements from all nodes, $b = \begin{bmatrix}b_1^T & b_2^T & \dots & b_p^T\end{bmatrix}^T$ and $\eta = \begin{bmatrix}\eta_1^T & \eta_2^T & \dots & \eta_p^T\end{bmatrix}^T$ are similarly defined. Recall that $\mathcal{L}$ is the Laplacian of the graph and $\otimes$ denotes the kronecker product. It is not hard to check that system~\eqref{eqn:globalestimate} falls into the class of systems of the form~\eqref{eqn:genestimate} by choosing $G=I_p\otimes \hat A$ all of whose eigenvalues have positive real parts, $F=F^T=\mathcal{L}\otimes I_n\succeq 0$, and $a=np$. We also choose $\varphi=\phi^1$ and we note that $ \phi_i^1[t+1]=\hat A\phi_i^1[t]$ with $i$ ranging from $1$ to $p$.

Next we show that all the additional Assumptions we made in this section are satisfied. 

First of all, we note that $G=I_p\otimes \hat A$ and $F=\mathcal{L}\otimes I_n$, therefore all eigenvalues of $G$ coincide with the eigenvalues of $A$ and all eigenvalues of $F$ coincide with the eigenvalues of $\mathcal{L}$. Moreover, since $\phi_i^1[t+1]=\hat A\phi_i^1[t]$ for any $i$, we have $\varphi[t+1]=(I_p\otimes \hat A)\varphi[t]=G\phi[t]$. To this end, we only need to check if Assumptions 1(b) and 3 are satisfied, which will be argued in the following.

\begin{enumerate}
    \item \textbf{Checking Additional Assumption 1(b):} We pick matrix $R=\frac{1}{\sqrt{p}}\mathbf{1}_p$ and a matrix $S$ of suitable dimension such that $\begin{bmatrix}R & S\end{bmatrix}$ is an orthogonal matrix. We note that the conditions $\mathcal{R}(R)=\mathrm{ker}F$ and $\mathcal{R}(S)\cap\mathrm{ker}F=\{0\}$ follow from the construction of $R$ and $S$. Also, recall that $\mathcal{L}$ is the Laplacian of the communication graph and the only eigenvector corresponding to the zero eigenvalue is $\mathbf{1}_p$. This, together with the construction of $R$, imply that the matrix $\begin{bmatrix}R & S\end{bmatrix}\otimes I_n$ is an orthogonal matrix with $R\otimes I_n$ lying in the kernel of $(\mathcal{L}\otimes I_n)$. Moreover, we can check by direct computation that the following equalities hold:
\begin{eqnarray*}
\begin{aligned}
    \hspace{1.8cm}(R^T\otimes I_n)G(S\otimes I_n)=(R^T\otimes I_n)(I_p\otimes \hat A)(S\otimes I_n) = R^TS\otimes \hat A=0, \\
    \hspace{1.8cm}(S^T\otimes I_n)G(R\otimes I_n)=(S^T\otimes I_n)(I_p\otimes \hat A)(R\otimes I_n) = S^TR\otimes \hat A=0.
\end{aligned}
\end{eqnarray*}
    
    \item \textbf{Checking Additional Assumption 3:} Since the graph is undirected, and hence $\mathcal{L}=\mathcal{L}^T$ holds, we can diagonalize  $\mathcal{L}$ with a matrix $T$, i.e.,  $T^{-1}\mathcal{L}T$ is a diagonal matrix with eigenvalues $\lambda_1,\dots,\lambda_r$. To show that the additional Assumption 3 is also satisfied, we directly compute the following two transformations: $(T\otimes I_n)^{-1}(\mathcal{L}\otimes I_n)(T\otimes I_n)=\text{diag}\{\lambda_1I_n,\dots,\lambda_rI_n\}$ and $(T\otimes I_n)^{-1}(I_p\otimes \hat A)(T\otimes I_n)=I_p\otimes \hat A=\text{diag}\{\hat A,\dots,\hat A\}$. This calculation shows that $F$ and $G$ are simultaneously diagonalizable and each block in $G$ corresponds to an identity block in $F$, after diagonalization, up to a scalar $\lambda_i$, hence the additional Assumption 3 is also satisfied.
    
\end{enumerate}

Since the state observer we designed in~\eqref{eqn:globalestimate} falls into the class~\eqref{eqn:genestimate} and satisfies all the additional assumptions, we can apply Propositions~\ref{prop:avg} and~\ref{prop:cons} to analyze the trajectory of $W_i^1$s in Equation~\eqref{eqn:localestimate}. By Proposition~\ref{prop:avg} we conclude that $(R^T\otimes I_n)W^1=\frac{1}{\sqrt p}\sum_iW_i^1$ converges to $(R^T\otimes I_n)\varphi=\frac{1}{\sqrt p}\sum_i\phi_i^1$ exponentially fast. By Proposition~\ref{prop:cons} we conclude that $(S^T\otimes I_n)W^1$ converges to $0$ exponentially fast, which shows the difference between any two estimates in two nodes decays to zero exponentially fast. Combining both we obtain that any $W_i$ converges to $\frac{1}{p}\sum_i\phi_i^1$ and hence Theorem~\ref{thm:main} is proved. 
\end{proof}

\subsection{Reconstructing the State with Compressed Measurements}

At this point, we know that if each node in the network executes the tracking algorithm \eqref{eqn:localestimate} corresponding to each set of weights $\{d_{ji}\}$, then its estimate vector $W_i$ converges to $(D\otimes I_n)Y$ exponentially fast. The missing piece now is a decoding algorithm which enables each node to reconstruct the state estimate $\hat x_i[t]$ from the compressed measurement $W_i[t]$. In the remainder of this section we argue that if $W_i[t]$ is close enough to $(D\otimes I_n)Y[t]$ then reconstructing the state $x[t]$ is nothing more than solving a standard SSR problem. Once again, for the sake of simplicity we are not going to provide the rigorous definition of the SSR problem in this paper, which can be found in~\cite{mao19}.

In the reconstruction step, each node seeks a state $\hat x\in\mathbb{R}^n$ and a vector $\hat E\in\mathbb{R}^{pn}$ containing no more than $s$ non-zero blocks that best explain the tracked compressed measurements $W_i$. This problem can be solved in different ways. In this paper we adopt an optimization formulation by requiring each node to minimize $\Vert W_i-(D\otimes I_n)(\mathcal{O} \hat x+\hat E)\Vert_2$. Note that we dropped the time index $t$ here to simplify notation. In other words, each node solves the following optimization problem:
\begin{eqnarray}
\label{eqn:opti} \notag
(\hat x,\hat E)=& \text{argmin}_{(\tilde x,\tilde E)} &\Vert W_i-(D\otimes I_n)(\mathcal{O} \tilde x+\tilde E)\Vert_2, \\ 
& s.t. &\Vert \tilde E\Vert_{l_0/ l_r}\leq s,
\end{eqnarray}
where $\Vert \tilde E\Vert_{l_0/ l_r}$ denotes the number of non-zero blocks of the vector $\tilde E$. The following lemma states that the solution $\hat x$ of \eqref{eqn:opti} is a good estimate of the system state when the tracking of the compressed measurements is accurate enough.

\begin{lemma}
\label{lemma:converge}

Let $D\in\mathbb{R}^{v\times p}$ and assume $(A,C)$ is $2s$-sparse detectable with respect to $D$. There exists a constant $\beta\in (0,+\infty)$ such that for any $\alpha \in (0,+\infty)$, $x\in\mathbb{R}^n$, $W_i\in\mathbb{R}^{vn}$, and $Y\in\mathbb{R}^{pn}$ satisfying $Y=\mathcal{O}x$ and $\Vert W_i- (D\otimes I_n)Y\Vert_2\leq \alpha$, the solution $\hat x$ to the optimization problem \eqref{eqn:opti} satisfies $\Vert\hat x- x\Vert_2\leq \beta\alpha.$
\end{lemma}

\begin{proof}
The solution $(\hat x, \hat E)$ to the optimization problem \eqref{eqn:opti} satisfies the following set of inequalities: 
\begin{eqnarray*}
&&\Vert W_i-(D\otimes I_n)(\mathcal{O} \hat x+\hat E)\Vert_2 \\
&\leq& \Vert W_i-(D\otimes I_n)(\mathcal{O} x+E)\Vert_2 \\
&\leq& \alpha,
\end{eqnarray*}
which implies:
\begin{equation}
\Vert (D\otimes I_n)(\mathcal{O}(\hat x-x)+(\hat E-E))\Vert_2\leq 2\alpha.
\end{equation}
Moreover, let $\mathcal{K}_1$ be the index set of non-zero blocks in $E$, $\mathcal{K}_2$ be the index set of non-zero blocks in $\hat E$, and $\mathcal{K}=\mathcal{K}_1\cup \mathcal{K}_2.$ Note that $|\mathcal{K}|\leq 2s$. We pick $L\in\mathbf{P}_{2s}$ such that\footnote{The definition of $\mathbf{P}_{2s}$ can be found in the definition of sparse detectability with respect to a matrix in Section \ref{sec:designD}, and the norm condition can be satisfied by a proper normalization.} $\mathrm{ker}(L)=\text{span}\{\mathbf{e}_i,i\in\mathcal{K}\}$ and $\Vert L\Vert_2=1$. We observe that the following inequality holds:
\begin{equation}
\label{eqn:lastconflict}
\Vert ((L\otimes I_n)(D\otimes I_n)(\mathcal{O}(\hat x-x))\Vert_2\leq 2\alpha.
\end{equation}
By assumption $(A,C)$ is $2s$-sparse detectable with respect to $D$, it follows that the matrix $(L\otimes I_n)(D\otimes I_n)\mathcal{O}$ has full column rank. In other words, we can pick $\beta^{-1}=\frac{1}{2}\sigma_{\min}((L\otimes I_n)(D\otimes I_n)\mathcal{O})>0$ and then Equation \eqref{eqn:lastconflict} directly implies $\Vert \hat x-x\Vert_2\leq \beta\alpha$, which finishes our proof. 

\end{proof}

As a special case, when $W_i=(D\otimes I_n)Y$, the optimization problem in \eqref{eqn:opti} degenerates to
the following equality: $$(D\otimes I_n)Y=(D\otimes I_n)\mathcal{O} \hat x_i+(D\otimes I_n)\hat E_i.$$ Solving this equality is equivalent to finding the state $\hat x_i$, a vector $T\in\mathrm{ker} (D\otimes I_n)$ and a vector $\hat E_i$ containing no more than $s$ non-zero blocks such that: $$Y=\mathcal{O}\hat x_i+\hat E_i+T,$$ holds. Now we define a matrix $N$ which has $np$ rows and the least possible number of columns such that $\mathcal{R}(N)=\mathrm{ker}(D\otimes{I_n})$, and the optimization problem \eqref{eqn:opti} is further reduced into finding the state $\hat x_i$, a vector $\hat r_i$ of appropriate dimension and a vector $\hat E_i$ containing no more than $s$ non-zero blocks satisfying: $$Y=\mathcal{O}\hat x_i+\hat E_i+N\hat r_i = \begin{bmatrix}\mathcal{O} & N \end{bmatrix}\begin{bmatrix}\hat x_i \\ \hat r_i\end{bmatrix}+\hat E_i.$$ Hereby we conclude that reconstructing the state $x$ from compressed measurements $(D\otimes I_n)Y$ is a special case of the SSR problem, and any algorithmic solution to the SSR problem, for example~\cite{shoukry2018smt} and \cite{mishra2016secure}, can be applied to solve the DSST problem if they do not require additional assumptions.

The preceding discussion directly leads to the sufficiency result:

\begin{lemma}
\label{thm:mainthm}
Consider the linear system subject to attacks defined in~\eqref{eqn:sys} satisfying Assumptions~\ref{assum2}-\ref{assum6} and a communication network satisfying Assumption \ref{assum1}. The tracking algorithm~\eqref{eqn:localestimate} together with a state-reconstruction algorithm enables each node to asymptotically track the state of the system~\eqref{eqn:sys} and consequently solves the DSST problem, if the pair $(A,C)$ is 2s-sparse detectable.
\end{lemma}

\section{Concluding Remarks}
\label{sec:conclusion}
In this paper, we studied the problem of tracking the state of a linear system against sensor attacks in a network. Motivated by recent works in the dynamic average consensus literature and the simple observation that a compressed version of measurements suffices to reconstruct the state, we proposed a novel decentralized observer that enables each node in the network to track the state of the linear system without making any non-necessary assumptions except one regarding the sampling rate.

Future directions include developing a decentralized state-tracking algorithm robust against other types of adversarial attacks, for example, a byzantine attack which is not only able to alter measurements but also forces attacked nodes to deviate from their prescribed algorithm. One additional problem of interest is the extension of the proposed results to the case where inputs are present but unknown to the nodes.

\bibliographystyle{IEEEtran}
\bibliography{references}

\end{document}